\def\ba{\begin{array}}
\def\ea{\end{array}}

\def\dps{\displaystyle}

\documentclass[aps,amsmath,amssymb,amsfonts,showpacs,twocolumn,pra]{revtex4}
\usepackage{graphicx}
\usepackage{caption2}
\usepackage{epstopdf}
\usepackage{braket}
\usepackage{graphicx}
\usepackage{amsmath}
\usepackage{amscd}
\usepackage{amsthm}
\usepackage{amssymb} \usepackage{latexsym}
\usepackage{eufrak}
\usepackage{euscript}
\usepackage{epsfig}
\usepackage{graphics}
\usepackage{array}
\usepackage{enumerate}
\usepackage{dsfont}
\usepackage{color}
\usepackage{wasysym}

\def\qed{\leavevmode\unskip\penalty9999 \hbox{}\nobreak\hfill
     \quad\hbox{\leavevmode  \hbox to.77778em{%
               \hfil\vrule   \vbox to.675em%
               {\hrule width.6em\vfil\hrule}\vrule\hfil}}
     \par\vskip3pt}

\newtheorem{lemma}{Lemma}

\newcommand{\R}{{\mathbb R}}
\theoremstyle{theorem}
\newtheorem{thm}{Theorem}[section]
\newcommand{\N}{{\mathbb N}}
\newcommand{\al}{\alpha}

\newcommand{\Q}{{\mathbb Q}}

\newcommand{\vol}{{\mathrm{vol}}}

\begin{document}

\title{Towards Grothendieck Constants and LHV Models in Quantum Mechanics}

\author{Bobo Hua$^{1,2}$}
\author{Ming Li$^{3}$}
\author{Tinggui Zhang$^{4}$}
\author{Chunqin Zhou$^{5}$}
\author{Xianqing Li-Jost$^{6}$}
\author{Shao-Ming Fei$^{6,7}$}

\vspace{4ex}

\affiliation{
$^1$ School of Mathematical Sciences, LMNS, Fudan University, Shanghai 200433, China\\
$^2$Shanghai Center for Mathematical Sciences, Fudan University, Shanghai 200433, China\\
$^{3}$ College of the Science, China University of Petroleum, 266580 Qingdao, China\\
$^4$ School of Mathematics and Statistics, Hainan Normal University, Haikou 571158, China\\
$^5$ Department of Mathematics, Shanghai Jiaotong University, Shanghai 200240, China\\
$^6$ Max-Planck-Institute for Mathematics in the Sciences, Leipzig 04103, Germany\\
$^7$ School of Mathematical Sciences, Capital Normal University, Beijing 100048, China
}

\begin{abstract}
We adopt a continuous model to estimate the Grothendieck
constants. An analytical formula to compute the lower bounds of
Grothendieck constants has been explicitly derived for arbitrary orders,
which improves previous bounds. Moreover, our lower bound of the Grothendieck constant of order three
gives a refined bound of the threshold value for the nonlocality of the two-qubit Werner states.
\end{abstract}


\maketitle

\section{Introduction}
Quantum mechanics exhibits the nonlocality of the nature in essence.
The impossibility of reproducing all correlations observed in
composite quantum systems using models \`a la
Einstein-Podolsky-Rosen (EPR) \cite{EinsteinPodolskyRosen35} was
proven in 1964 by Bell \cite{Bell64}. A quantum state is said to admit a local
hidden variable (LHV) model if all the measurement outcomes
can be modeled as a classical random distribution over
a probability space. Consider a bipartite state $\rho$
in $\mathcal{H}_\texttt{A}\otimes\mathcal{H}_\texttt{B}$ with subsystems $\texttt{A}$ and $\texttt{B}$.
If Alice performs a measurement $A$ on the subsystem $\texttt{A}$ with an outcome $a_i$ and, at
space-like separation, Bob performs a measurement $B$ on the subsystem $\texttt{B}$ with an outcome
$b_j$, then an LHV model supposes that the joint probability of
getting $a_i$ and $b_j$ satisfies
$$
\mathrm{Pr}(a_i,b_j|A,B,\rho)=\int_{\Omega}\mathrm{Pr}(a_i|A,\lambda)\,\mathrm{Pr}(b_j|B,\lambda)\,d\omega^{\rho}(\lambda),
$$
where $d\omega^{\rho}(\lambda)$ is some
distribution over a space $\Omega$ of hidden variable $\lambda$.
A quantum state is called \emph{local} if it admits an LHV model, and
\emph{nonlocal} otherwise.

In his seminal work, Bell showed that all quantum states
admitting LHV models satisfy the so-called Bell inequalities \cite{Bell64}.
That is, a state admits no LHV models if it violates some Bell
inequalities. It is known that every pure entangled bipartite or multipartite state
violates a generalized Bell inequality \cite{GisinPeres92,PopescuRohrlich92}.
Namely, for pure states the entanglement and the non-locality coincide.
However for mixed states, the situation is more complicated.
There are no general methods to judge whether a mixed state admits an LHV model or not, i.e. to find all Bell inequalities
is computationally hard \cite{Pitowsky,AlonNaor}. Even for the most concerning two-qubit Werner states,
the precise threshold value of nonlocality is still unknown.

As the ``fundamental theorem in the metric theory of tensor products", the Grothendieck's theorem
\cite{Grothendieck53} had a major impact on Banach space
theory. The constants related to the
Grothendieck's theorem are nowadays called Grothendieck constants \cite{Pisier12}.
It turns out that the Grothendieck constants are related to the Bell inequalities, observed by Tsirelson \cite{Tsirelson87}.
Of particular interest, Acin-Gisin-Toner \cite{AcinGisinToner06} demonstrated that the threshold value for the
non-locality of the Werner states for projective measurements is given explicitly by the
Grothendieck constant of order three, $K(3),$ see Section \ref{s:Grothendieck} for the definition. This reduces the
problem of the nonlocality of Werner states to the computation of the exact value of $K(3).$
However, it is formidably difficult to compute the Grothendieck constants except for the case of order two.
Generally, what one can do is to estimate the lower and upper bounds of the Grothendieck constants.

In this paper, by generalizing the Bell operator with 465 measurement settings on each side in \cite{Vertesi08} to
a continuous model with infinitely many measurement settings, we present
an analytical formula in estimating the lower bounds of the
Grothendieck constants. This formula is valid for Grothendieck constants of arbitrary order and
improves many previously obtained bounds. From our
lower bound of $K(3)$, we
derive a bound of the threshold value for the nonlocality of the two-qubit Werner states, which
gives the best knowledge about such nonlocality up to date.

\section{Lower bound of Grothendieck constant for arbitrary order $d$}\label{s:Grothendieck}

Let $\mathcal{M}_m(\R)$ be the set of $m\times m$ real matrices and
$S^{d-1}$ the unit sphere in $\R^d$, $m,d\in \N$. Given $M\in
\mathcal{M}_m(\R),$ we define
\begin{equation}\label{e:classical}
C(M)=\sup\left|\sum_{i,j=1}^mM_{ij}\,a_ib_j\right|,
\end{equation}
where the supremum is taken over all possible assignment $a_i,b_j\in
\{1,-1\},$ $1\leq i,j\leq m$. Replacing $a_i,b_j$ by $d$-dimensional unit
vectors, we define
$$
Q(M)=\sup_{\mathbf{a_i},\mathbf{b_j}\in S^{d-1}}\left|\sum_{i,j=1}^mM_{ij}\,\mathbf{a_i}\cdot\mathbf{b_j}\right|,
$$
where the supremum is taken over all $d$-dimensional unit
vectors $\mathbf{a_i}$ and $\mathbf{b_j}$ and $\mathbf{a_i}\cdot\mathbf{b_j}$ denotes their scalar product.
The Grothendieck constant of order $d$ is defined by
\begin{equation}\label{G}
K(d)=\sup_{m\geq 1}\,\sup_{\substack{M\in\mathcal{M}_m(\R)\\M\neq 0}}\frac{Q(M)}{C(M)}.
\end{equation}

The value of $C(M)$ depends only on the choice of the matrix $M$. Besides $M$,
$Q(M)$ also depends on the dimension $d$ of Euclidean space $\R^d$ where we choose unit
vectors $\mathbf{a_i}$ and $\mathbf{b_j}$.
It is a great challenge to evaluate the Grothendieck constant $K(d)$ for general $d$.
Till now the only exactly known result of $K(d)$ is for $d=2$, $K(2)=\sqrt{2}$ \cite{Krivine79}.
For $d\geq 3$, there are some lower bounds of $K(d)$: For instance,
Bri\"et-Buhrman-Toner \cite{BrietBuhrmanToner11} obtained a lower bound of $K(d)$ for general $d$,
$$
K(d)\geq \frac{\pi}{d}\left(\frac{\Gamma(\frac{d+1}{2})}{\Gamma(\frac{d}{2})}\right)^2.
$$

In this section, we propose a continuous model to compute the lower
bounds of the Grothendieck constant for arbitrary $d$.
Our results improve some known lower bounds. In particular, the lower bound of $K(3)$ we obtained is the best up to date,
which improves the result on the non-locality of the Werner states.

For any $n\in \N,$ let $m=n+n(n-1)/2.$ We choose the following
special $M\in M_m(\R)$ such that \cite{Vertesi08}
$$
C(M)=\sup\left|\sum_{i,j=1}^na_ib_j+\sum_{i<j}[\al_{ij}(b_i-b_j)+\beta_{ij}(a_i-a_j)]\right|,
$$
where the supremum is taken over $a_i,\alpha_{ij}\in \{1,-1\}$ and
$b_j,\beta_{ij}\in \{1,-1\}$ ($\{a_i,\alpha_{ij}\}$ stands for
$\{a_i\}_{i=1}^m$ in the definition \eqref{e:classical}, similarly $\{b_i,\beta_{ij}\}$ for $\{b_j\}_{j=1}^m$). Hence by
choosing $\al_{ij}=\mathrm{sign}(b_i-b_j)$ and $\beta_{ij}=\mathrm{sign}(a_i-a_j)$, one has
$$
C(M)=\sup\left|\sum_{i,j=1}^na_ib_j+\sum_{i<j}(|b_i-b_j|+|a_i-a_j|)\right|=n^2.
$$

Correspondingly, for given $M$, to get a better bound of $Q(M)$ one needs to suitably
choose the vectors $\mathbf{a_i}$ and $\mathbf{b_j}$.
By setting vectors $\mathbf{\al_{ij}}$ ($\mathbf{\beta_{ij}}$ resp.) parallel to
$\mathbf{b_i}-\mathbf{b_j}$ ($\mathbf{a_i}-\mathbf{a_j}$ resp.)
and $\mathbf{a_i}=\mathbf{b_i},$ $1\leq i\leq n$, one obtains
\begin{equation}
Q(M)=\sup_{\mathbf{a_i}\in S^{d-1}}\left(|\sum_i
\mathbf{a_i}|^2+2\sum_{i<j}|\mathbf{a_i}-\mathbf{a_j}|\right).
\end{equation}
Therefore
\begin{equation}\label{e:discrete version}K(d)\geq
\sup_{\mathbf{a_i}\in S^{d-1}}\left(\left|\frac{1}{n}\sum_i
\mathbf{a_i}\right|^2+\frac{1}{n^2}\sum_{i\neq
j}\left|\mathbf{a_i}-\mathbf{a_j}\right|\right),
\end{equation}
where the supremum is taken over $(\mathbf{a_1},\cdots,\mathbf{a_n})\in (S^{d-1})^n$
for any $n\geq 1$.

Using this discrete version of optimization problem, V\'ertesi \cite{Vertesi08} runs
a numerical simulation to obtain the lower bounds of $K(d),$ $3\leq
d\leq 5.$ Since the complexity of the computation grows
exponentially, the computation of the lower bounds becomes impossible
for large $d.$ Instead of this, we propose a continuous model to
reformulate this problem. Let $P(S^{d-1})$ denote the space of
probability measures on $S^{d-1}.$ Then we have

\begin{thm} For $d\geq 1,$
\begin{equation}\label{e:measures}
\ba{rcl}
K(d)&\geq& \sup_{\mu\in P(S^{d-1})}\left(\left|\int_{S^{d-1}}\mathbf{x}\,d\mu(\mathbf{x})\right|^2\right.\\[5mm]
&&\dps\left.+\int_{S^{d-1}\times S^{d-1}}
|\mathbf{x}-\mathbf{y}|\,d\mu(\mathbf{x})d\mu(\mathbf{y})\right).
\ea
\end{equation}
\end{thm}

\begin{proof}
From the discrete version \eqref{e:discrete version}, the assertion
\eqref{e:measures} holds for rational convex combination of delta
measures, i.e. $\mu=\sum_{i=1}^N \lambda_i \delta_{\mathbf{a_i}}$ with
$\lambda_i\in \Q_+,$ $\sum_{i=1}^N\lambda_i=1$ and $\mathbf{a_i}\in S^{d-1}$.
Note that any probability measure $\mu\in
P(S^{d-1})$ can be approximated by convex combination of delta
measures in the weak topology (precisely, weak* topology in the terminology of functional analysis).
This follows from Krein-Milman theorem and the fact that delta measures are extreme points for the set of
probability measures \cite{referee}. Since $|\mathbf{x}|,|\mathbf{x}-\mathbf{y}|$ are
continuous functions, the theorem follows from the definition of weak
convergence of measures.
\end{proof}

Hence, to derive an effective lower bound of $K(d)$ it suffices to
choose some good measures for this optimization problem. It seems that
the problem becomes more complicated since the finite
dimensional optimization problem has been transformed to an infinite dimensional problem on
the space of measures. However, this is in fact an advantage which allows one to choose nice
absolutely continuous measures on a sphere to get explicit lower
bounds of the Grothendieck constants.

Let $(\phi_1,\phi_2,\cdots,\phi_{d-1})$ be the spherical coordinates of $S^{d-1}\subset \R^{d}$ such that
\[\left\{\begin{array}{l}
x_1=\sin{\phi_1}\sin{\phi_2}\cdots\sin(\phi_{d-1}),\\
x_2=\sin{\phi_1}\sin{\phi_2}\cdots\cos(\phi_{d-1}),\\
\cdots\\
x_{d-1}=\sin{\phi_1}\cos{\phi_2},\\
x_d=\cos{\phi_1},
\end{array}\right.\] where $\phi_i\in [0,\pi]$ for $1\leq i\leq d-2$ and $\phi_{d-1}\in [0,2\pi).$
We denote by $d\mu=\sin^{d-2}\phi_1\sin^{d-3}\phi_{2}\cdots\sin\phi_{d-2}\,d\phi_1d\phi_2\cdots d\phi_{d-1}$ the
spherical (volume) measure of $S^{d-1}$. For simplicity, we also denote by $\mathbf{x}=(\mathbf{x}',x_{d}),$ $\mathbf{x}'\in \R^{d-1}$,
the Cartesian coordinates of $\mathbf{x}\in S^{d-1},$  by $P_d(\mathbf{x})=x_d$
the projection to the $d$-th coordinate, and by $\phi_1(\mathbf{x})$ the first
spherical coordinate of $\mathbf{x}$. For any $0\leq a\leq b\leq
\frac{\pi}{2},$ we denote
$$\Omega_{a}^b=\{\mathbf{x}\in S^{d-1}: a\leq \phi_1(\mathbf{x})\leq b\},$$
which is contained in the upper hemisphere $S^{d-1}\cap \{x_d\geq 0\}.$
To obtain the lower bounds of $K(d),$ we will choose the uniform probability measure
on $\Omega_a^b,$ i.e.
$\mu_{a,b}=\frac{1}{\vol\,{\Omega_a^b}}\vol|_{\Omega_a^b}.$ The variables
$a,b$ are introduced to refine (i.e. maximize in some sense) the lower bound of $K(d)$,
since we have no priori knowledge of the optimal measure which
attains the maximum on the right hand side of \eqref{e:measures}.

Then second term of the right hand side of \eqref{e:measures} corresponds to
$$
\ba{l}\dps
\int_{\Omega_{a}^b\times\Omega_{a}^b}|\mathbf{x}
-\mathbf{y}|d\mu_{a,b}(\mathbf{x})d\mu_{a,b}(\mathbf{y})\\
~~~~~~~~~~~\dps=\frac{1}{\vol{(\Omega_a^b})^2}
\int_{\Omega_{a}^b\times\Omega_{a}^b}|\mathbf{x}-\mathbf{y}|d\mu(\mathbf{x})d\mu(\mathbf{y}),
\ea
$$
which involves a $2(d-1)$-multiple integration. This is
the obstruction for the numerical computation of the lower bounds for large $d.$
Nevertheless, the computation can be considerably simplified
by the symmetry of the sphere.
\begin{lemma}
For any $\mathbf{x},\,\tilde{\mathbf{x}}\in S^{d-1}$ satisfying $\phi_1(\mathbf{x})=\phi_1(\tilde{\mathbf{x}}),$
we have
\begin{equation}\label{e:1}
\int_{\Omega_a^b}|\mathbf{x}-\mathbf{y}|d\mu(\mathbf{y})=\int_{\Omega_a^b}|\tilde{\mathbf{x}}-\mathbf{y}|d\mu(\mathbf{y}).
\end{equation}
\end{lemma}
\begin{proof}
For given $\mathbf{x}$ and $\tilde{\mathbf{x}}$ with $P_d(\mathbf{x})=P_d(\tilde{\mathbf{x}}),$
there is an isometry of $S^{d-1}$, $A\in SO(d),$ such that $A(\mathbf{x})=\tilde{\mathbf{x}}$ and
$A(\Omega_a^b)=\Omega_a^b.$ This can be seen as follows: Without loss of generality, one may write $\mathbf{x}=(\mathbf{x}',c)$
and $\tilde{\mathbf{x}}=(\tilde{\mathbf{x}}',c)$ where $c=P_d(\mathbf{x}).$ Since $\mathbf{x}',\tilde{\mathbf{x}}'\in
\sqrt{1-c^2}\,S^{d-2},$ there is a $B\in SO(d-1)$ such that
$B\mathbf{x}'=\tilde{\mathbf{x}}'.$ Hence one can take
\[A=\begin{pmatrix}
  B & 0  \\
  0 & 1
 \end{pmatrix}.\]
Since $SO(d)$ acts isometrically on $S^{d-1},$ the equation \eqref{e:1} follows.
\end{proof}

By the Lemma, for given $\mathbf{x}\in \Omega_a^b$ the integral over
$\mathbf{y}$ only depends on $\phi_1(\mathbf{x}).$ Without loss of generality, we may choose
$\mathbf{x}=(0,0,\cdots,\sin{\phi_1},\cos{\phi_1}),$ i.e. its spherical
coordinate is $(\phi_1,0,\cdots,0)$. Then for any $\mathbf{y}$ whose
spherical coordinate reads $(\psi_1,\psi_2,\cdots,\psi_{d-1}),$
one has
\begin{equation}\label{xmy}
|\mathbf{x}-\mathbf{y}|=\sqrt{2-2\sin{\phi_1}\sin{\psi_1}\cos{\psi_2}-2\cos{\phi_1}\cos{\psi_1}}.
\end{equation}
Therefore we obtain
$$
\ba{l}
\dps\int_{\Omega_{a}^b\times\Omega_{a}^b}|\mathbf{x}-\mathbf{y}|d\mu(\mathbf{x})d\mu(\mathbf{y})\\[4mm]
\dps=\vol(S^{2-3})\vol(S^{d-3})\int_a^b\int_a^b\int_0^\pi f(\phi_1,\psi_1,\psi_2)d\phi_1d\psi_1d\psi_2
\ea
$$
for $d\geq 3$, where
\begin{equation}\label{e:f-function}
f(\phi_1,\psi_1,\psi_2)=|\mathbf{x}-\mathbf{y}|\sin^{d-2}\phi_1\sin^{d-2}{\psi_1}\sin^{d-3}\psi_2,
\end{equation}
with $|\mathbf{x}-\mathbf{y}|$ given by (\ref{xmy}).
Combining above results, we have the following theorem:

\begin{thm}\label{thm:main}
The Grothendieck constant $K(d)$, $d\geq3$, satisfies
\begin{widetext}
\begin{equation}\label{e:explicit formula}
K(d)\geq\frac{1}{\left(\int_a^b\sin^{d-2}\phi_1d\phi_1\right)^2}\left[\left(\left.
\frac{\sin^{d-1}\phi_1}{d-1}\right|_a^b\right)^2+\frac{\vol(S^{d-3})}
{\vol(S^{d-2})}\int_a^b\int_a^b\int_0^\pi f(\phi_1,\psi_1,\psi_2) d\phi_1d\psi_1d\psi_2\right],
\end{equation}
\end{widetext}
where $f(\phi_1,\psi_1,\psi_2)$ is given by \eqref{xmy} and \eqref{e:f-function}.
\end{thm}

By reducing a $2(d-1)$-multiple integration to a triple integration for any $d\geq3$,
we have obtained a lower bound of the Grothendieck constants which can be easily calculated via numerical methods.
By varying $a,b$ in the domain $\{(a,b):0\leq a<b\leq \frac{\pi}{2}\},$ one may get a refined lower bound of $K(d)$ by
maximizing the right hand of \eqref{e:explicit formula} for these $a,b$.
For instance, by taking $a=0$, $b=1.04819755$, one gets $K(3)\geq 1.41758$.
Taking $a=0.742832$, $b=0.749115$, one gets $K(5)\geq 1.46112$.
Some numerical results are listed in the following table:
\begin{center}
    \begin{tabular}{| l | l|  l |  l|}
     \hline
  & \multicolumn{3}{|c|}{$K(d)\geq$} \\
    \hline
    d& V\'ertesi  &Bri\"et etc.   & our result \\ \hline
    3 &1.41724 &1.33333 &1.41758 \\ \hline
    4 &1.44521&1.38791  &1.44566 \\ \hline
    5 &1.46007 &1.42222  &1.46112 \\ \hline
    6 &&1.44574 &1.47017 \\ \hline
    7 &&1.46286  &1.47583 \\ \hline
    8 &&1.47586  & 1.47972  \\ \hline
    9 &&1.48608 &1.48254\\ \hline
    \end{tabular}
\end{center}
The first column collects the results of V\'ertesi \cite{Vertesi08} for the case $d=3,4,5.$
The second column contains the lower bounds of $K(d)$ proved by Bri\"et-Buhrman-Toner \cite{BrietBuhrmanToner11}.
From this table, one immediately figures out that our results improve the results of V\'ertesi \cite{Vertesi08} and
Bri\"et-Buhrman-Toner \cite{BrietBuhrmanToner11} for $d\leq 8$.

\section{Non-locality of two-qubit Werner states}

In 1989, Werner explicitly constructed LHV models for some entangled mixed bipartite states\cite{Werner89}.
The two-qubit Werner state is given by
\begin{equation}\label{werner}
\rho^W_p=p\Ket{\psi^-}\Bra{\psi^-}+(1-p)\mathds{I}/4,
\end{equation}
where in computational basis,
$\Ket{\psi^-}=(\Ket{01}-\Ket{10})/\sqrt2$ and $0\leq p\leq 1.$
$\rho^W_p$ is separable if and only if $p\leq 1/3$ \cite{Werner89}.
It admits an LHV model for all measurements for $p\leq 5/12$ \cite{Barrett02},
and admits an LHV model for projective measurements for $p\leq 0.6595$ \cite{AcinGisinToner06}.

Let $A_{i}$ and $B_{i}$, $i=1,2,...,m$, be dichotomic observables with respect to the
two qubits, $A_i=\mathbf{a}_i\cdot\mathbf{\sigma}$ and
$B_i=\mathbf{b}_i\cdot\mathbf{\sigma}$, with
$\mathbf{\sigma}=({\sigma_1},{\sigma_2},{\sigma_3})$
the Pauli matrices. $\mathbf{a}_i=(a^{(1)}_i,a^{(2)}_i,a^{(3)}_i)$,
$\mathbf{b}_i=(b^{(1)}_i,b^{(2)}_i,b^{(3)}_i)$ are 3-dimensional real unit vectors.
For any Bell operator,
$$
B(M)=\sum_{i,j=1}^m M_{ij}\,{A_i}\otimes {B_j},
$$
where $M\in \mathcal{M}_m(\R)$ as in (\ref{e:classical}),
the mean value is given by
$$
Tr(B(M)\rho^W_p)=p\sum_{i,j=1}^mM_{ij}\,\mathbf{a_i}\cdot\mathbf{b_j}.
$$
Therefore the maximal violation of the corresponding Bell inequality is given by $p\, K(3)$ and
$\rho^W_p$ admits LHV models for projective measurements if and only if $p\leq 1/K(3)$ \cite{AcinGisinToner06}.
Hence the nonlocality problem of the two-qubit Werner states is reduced to estimate the value of
$K(3)$.

However, the precise value of $K(3)$ is still unknown.
There are various attempts to derive the upper and lower bounds of the
Grothendieck constants. For instance, Krivine \cite{Krivine79} showed that $K(3)\leq 1.5163$.
The Clauser-Horne-Shimony-Holt (CHSH)
inequality implies that $K(3)\geq \sqrt2,$ \cite{ClauserHorneShimonyHolt69}.
V\'ertesi \cite{Vertesi08} constructed Bell inequalities involving
465 settings on each qubit ($\{A_i,B_i\}_{i=1}^{m},\ m\geq 465$) to
show that $K(3)\geq 1.417241$, i.e., $\rho^W_p$ admits no LHV models for $p>0.705596$.
From Theorem \ref{thm:main}, we have shown that $K(3)\geq 1.41758$.
Therefore $\rho^W_p$ admit no LHV models for $p>0.705428$.
This provides the best known bound for the nonlocality of two-qubit Werner states. (see Fig. 1:
Acin-Gisin-Toner show that $\rho^W_p$ admits LHV models for $p\leq 0.66$.
Vertesi shows that $\rho^W_p$ admits no LHV models for $p> 0.705596$.
We show that $\rho^W_p$ admits no LHV models for $p> 0.705428$.)

\begin{figure}[h]
\resizebox{8.5cm}{!}{\includegraphics{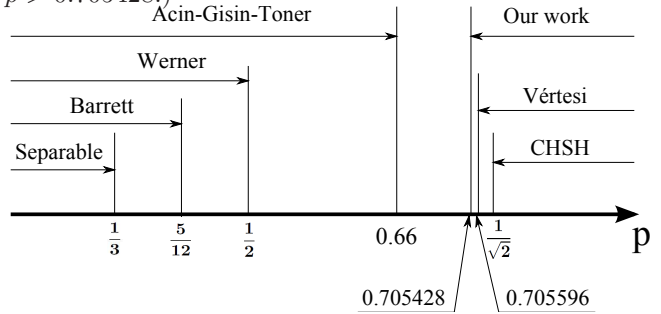}}\\
\caption{Nonlocal properties of two-qubit Werner states.}
\end{figure}

\section{Conclusion and discussions}

We have presented an analytical formula to estimate the lower bounds of the
Grothendieck constants for arbitrary order. It has been shown that
our bounds improve the previously obtained bounds for $d=3,...,8$.
It is also straightforward to calculate the lower bounds for higher order $d$. However, the lower bounds of $K(d)$
for large $d$ derived by this approach cannot exceed $1.5$,
which has been proven in \cite{Vertesi08}, section IV.
This certainly doesn't beat the best lower bound of 1.677 for $K(\infty)$ in \cite{DavieReeds}.

Nevertheless, our lower bound of $K(3)$ gives
a bound of the best threshold value for the nonlocality of the two-qubit Werner states up to date.
In fact, our new lower bounds of the Grothendieck constant of high orders can be also
used to improve the knowledge about the non-locality for higher dimensional
quantum states such that the related mean values of Bell operators are determined by the
Grothendieck constants \cite{AcinGisinToner06}.

%The quantum discord \cite{ollivier,Vedral} has been introduced
%in characterizing information correlations between two subsystems, which is the
%minimal amount of mutual information that can not be learned by any measurements only on one of the subsystems.
%The quantum discord is zero for all classical-classically correlated states and greater than zero
%for all quantum-classically (classical-quantum) or quantum-quantum correlated states.
%The quantum entanglement \cite{Nielsen,Horodecki09,Guhne09} has been introduced to certify if the measurements on one subsystem of a
%bipartite or multipartite state would affect the measurement outcomes from other subsystems.
%The entanglement measure, e.g. entanglement of formation and concurrence, gets zero for separable states
%and greater than zero for all entangled states.
%These quantities such as quantum discord, entanglement of formation and concurrence are well defined, though formidably
%difficult to compute analytically in general.
So far there are no effective ways to justify whether a quantum entangled state admits local hidden variable models or not.
For two-qubit Werner states, the Grothendieck constant $K(3)$ plays the essential role
that $\rho^W_p$ admits LHV models if and only if $p\leq K(3)^{-1}$.
It would be also interesting to find such quantities for general quantum states.

\bigskip
\noindent{\bf Acknowledgments}\, \, We would like to thank Kai Chen for introducing the problem and
useful discussions. This work was done in Max-Planck Institute for
Mathematics in the Sciences and supported by the NSFC 11275131, 11105226, 11271253 and 11401106;
the Fundamental Research Funds for the Central Universities No.12CX04079A, No.24720122013;
Research Award Fund for outstanding young scientists of Shandong Province No.BS2012DX045.

\end{document}